\documentclass[birkjour, showpacs,preprintnumbers,amsmath,amssymb]{revtex4}
\usepackage{graphicx}% Include figure files
\usepackage{bm}
\usepackage{epstopdf}
\usepackage{float}
\usepackage{dcolumn}% Align table columns on decimal point
\usepackage{bm}% bold math
\usepackage{mathrsfs}
\usepackage{amsmath}
\newtheorem{theorem}{Theorem}[section]

\newenvironment{proof}[1][Proof.]{\begin{trivlist}
     \item[\hskip \labelsep {\bfseries #1}]}{\end{trivlist}}

\begin{document}

%\preprint{APS/123-QED}

\title{Potentials of the Heun class: the triconfluent case}% Force line breaks with \\

\author{D. Batic}
\email{davide.batic@uwimona.edu.jm}
\affiliation{%
Department of Mathematics,\\  University of the West Indies, Kingston 6, Jamaica 
}
\author{D. Mills-Howell}
\email{dominic.millz27@gmail.com}
\affiliation{%
Department of Mathematics,\\  University of the West Indies, Kingston 6, Jamaica 
}
\author{M. Nowakowski}
\email{mnowakos@uniandes.edu.co}
\affiliation{
Departamento de Fisica,\\ Universidad de los Andes, Cra.1E
No.18A-10, Bogota, Colombia
}%

\date{\today}% It is always \today, today,

\begin{abstract}
Since the advent of quantum mechanics different approaches
to find analytical solutions of the Schr\"odinger equation have been
successfully developed. Here we follow and generalize
the approach pioneered by Natanzon and others by which the
Schr\"odinger equations can be transformed
into another well-known equation for transcendental function
(e.g., the hypergeometric equation). 
This sets a class of potentials for which this transformation is
possible.
Our generalization consists
in finding potentials allowing the transformation
of the Schr\"odinger equation into
a triconfluent Heun equation. We find the energy
eigenvalues of this class of potentials, the eigenfunction
and the exact superpartners. 
\end{abstract}

\pacs{XXX}% PACS, the Physics and Astronomy
\maketitle

\section{Introduction}
The study of exactly solvable Schr\"{o}dinger equations can be traced back to the beginning of quantum mechanics. The earliest examples 
are represented by the harmonic oscillator, Coulomb, Morse, P\"{o}schl-Teller, Eckart and Manning-Rosen potentials \cite{Fl,Derezinski,Lam}. By 
looking closer at these examples one can start to conjecture that exact solvability of the Schr\"{o}dinger equation depends on the fact 
that such an equation can be suitably reduced to the hypergeometric or the confluent hypergeometric equation. The problem of deriving the 
most general class of potentials such that the Schr\"{o}dinger equation can be transformed into the hypergeometric equation has been 
solved by \cite{Natanzon}. Furthermore, \cite{Natanzon1} studied solutions regular at infinity for the basic SUSY ladder of Hyperbolic P\"{o}schl-Teller potentials that admit representations in terms of confluent Heun polynomials. \cite{Davide2013} derived new classes of potentials such that the one-dimensional Schr\"{o}dinger equation can 
be turned into the Heun equation and its confluent cases. The generalized Heun equation has been considered as well in \cite{Davide2013}. 
Since the hypergeometric equation 
\begin{equation}\label{dueI}
y(1-y)\frac{d^2 v}{dy^2}+\left[c-(a+b+1)y\right]\frac{dv}{dy}-abv(y)=0,\quad y\in I\subset\mathbb{R}
\end{equation}
with $a,b,c\in\mathbb{R}$ is a special case of the Heun equation 
\begin{equation}\label{HeunH}
\frac{d^2 v}{dy^2}+\left(\frac{\gamma}{y}+\frac{\delta}{y-1}+\frac{\epsilon}{y-a}\right)\frac{dv}{dy}
+\frac{\alpha\beta y-q}{y(y-1)(y-a)}v(y)=0,\quad y\in I\subset\mathbb{R},
\end{equation}
whenever $q=\alpha\beta a$ and $\epsilon=0$, the potential classes obtained in \cite{Davide2013} generalizes the Natanzon's class. We 
start by reviewing the method developed by \cite{Milson,Davide2013} allowing the construction of the most general potential such that the Schr\"{o}dinger 
equation can be reduced to the triconfluent Heun equation (\ref{TCHeq}). Such a potential contains six free parameters.  

\section{Reduction of the radial Schr\"{o}dinger equation to a triconfluent Heun equation}\label{red}
We consider a quantum particle in a central field in three and two spatial dimensions, respectively. In the three dimensional case the 
behaviour of the particle is described by the Schr\"{o}dinger equation ($\hbar^2=2m=1$)
\[
i\partial_t\Psi(t,{\bf{r}})=H\Psi(t,{\bf{r}}),\quad
H=-\Delta+U(r)
\]
with $\Psi\in L^2(\mathbb{R}^3)$. Here, $\Delta$ denotes the Laplacian. Since the Hamilton operator commutes with the angular momentum operator 
${\bf{L}}=(L_x,L_y,L_z)$, it is sufficient to solve the eigenvalue problem for $H$, i.e. the time-dependent Schr\"{o}dinger equation, on the 
subspace of $L^2(\mathbb{R}^3)$ spanned by a basis of common eigenvectors of $|{\bf{L}}^2|$ and $L_z$. Such a subspace is made of functions having 
the form $\Phi_{\ell m}(r,\theta,\phi)=\Phi(r) Y_{\ell m}(\theta,\phi)$ where $\Phi\in L^2(\mathbb{R}_{+},r^2 dr)$ and $Y_{\ell m}(\theta,\phi)$ 
denotes the spherical harmonics. In spherical coordinates the Laplace operator becomes
\[
\Delta=\frac{1}{r^2}\partial_r r^2\partial_r-\frac{|{\bf{L}}^2|}{r^2}.
\]
If $\Phi_{\ell m}$ is an eigenfunction of $H$ relative to the eigenvalue $E$, then $\Phi$ must satisfy the radial Schr\"{o}dinger equation 
\cite{messiah}
\[
\left[-\left(\frac{d^2}{dr^2}+\frac{2}{r}\frac{d}{dr}-\frac{\ell(\ell+1)}{r^2}\right)+U(r)\right]\Phi(r)=
E\Phi(r). 
\]
The above equation can be further simplified if we bring it to its canonical form by means of the transformation $\Phi(r)=\psi(r)/r$ with 
$\psi\in L^2(\mathbb{R}_{+},r^2 dr)$ and we obtain
\begin{equation}\label{SH1} 
\frac{d^2\psi}{dr^2}+\left[E-V_{eff}(r)\right]\psi(r)=0,\quad
V_{eff}(r)=\frac{\ell(\ell+1)}{r^2}+U(r).
\end{equation}
The region where the dynamics of the particle can take place is constrained to the half axis $r>0$. Moreover, it will be assumed that the effective 
potential  does not depend on the energy of the particle. Two-dimensional quantum systems appear in solid state physics in connection with the 
fractional quantum Hall effect \cite{Zhang} and high temperature superconductivity \cite{Fetter,Wiegmann,Polyakov}. For the two-dimensional case 
we consider a system of two anyons \cite{Lerda}, i.e. particles with fractional statistics, in a spherically-symmetric potential. The 
time-independent Schr\"{o}dinger equation governing this system is \cite{Manuel,Roy} ($\hbar^2=2\mu=1$ with $\mu$ the reduced mass of the system)
\[
\left[-\left(\frac{d^2}{dr^2}+\frac{1}{r}\frac{d}{dr}-\frac{\nu^2}{r^2}\right)+U(r)\right]\Phi(r)=E\Phi(r),\quad
\nu=m+\frac{\theta}{\pi}, 
\]
where $\theta$ is the statistics parameter, $m\in\mathbb{Z}$ is the azimuthal quantum number and $r=|{\bf{r}}_1-{\bf{r}}_2|>0$ is the relative 
coordinate. Note that $r\neq 0$ to avoid that the two particles overlap. By means of the substitution $\Phi(r)=\psi(r)/\sqrt{r}$ we get formally 
the same Schr\"{o}dinger equation as given by (\ref{SH1}) with 
\begin{equation}\label{anyon}
V_{eff}(r)=\frac{\nu^2-\frac{1}{4}}{r^2}+U(r). 
\end{equation}
We want to construct the most general potential $U$ such that the radial Schr\"{o}dinger equation (\ref{SH1}) with effective potential given by 
(\ref{SH1}) with $\ell=0$ or (\ref{anyon}) with fixed $\nu$ can be transformed into the triconfluent Heun equation \cite{Ronveaux}
\begin{equation}\label{TCHeq}
\frac{d^2v}{d\rho^2}+I(\rho)v(\rho) = 0,\quad
I(\rho)=A_0+A_1\rho+A_2\rho^2-\frac{9}{4}\rho^4
\end{equation}
with parameters $A_1,A_2,A_2\in\mathbb{R}$ and $\rho\in\Omega\subseteq\mathbb{R}$. Note that (\ref{TCHeq}) can be obtained by means of a confluence process of the singularities involved 
in the Heun equation. To appreciate the role of the TCH equation in physics, we refer to \cite{Hioe,Liang,Schulze,Bay} where this equation 
appears in the treatment of anharmonic oscillators in quantum mechanics. A short exposition of known results concerning anharmonic oscillators can be found in Ch. I of \cite{ush}. It is interesting to observe that the radial Schr\"{o}dinger equation 
(\ref{SH1}) reduces immediately to the above triconfluent Heun equation whenever  
\[
r=\rho,\quad \psi=v,\quad E=A_0,\quad\ell=0,\quad U(r)=-A_1 r-A_2 r^2+\frac{9}{4}r^4,
\]
i.e. $U$ is a quartic potential. In what follows we are interested in non trivial transformations of the dependent variable $\psi$ and the radial 
coordinate $r$ transforming (\ref{SH1}) into (\ref{TCHeq}). As in \cite{Davide2013} we introduce the coordinate transformation $\rho=\rho(r)$ 
with $r>0$ and (\ref{TCHeq}) becomes 
\[
\frac{d^2 v}{dr^2} - \frac{\rho^{''}}{\rho^{'}}\frac{dv}{dr} +(\rho^{'})^2 I(\rho(r))v(\rho(r)) = 0,\quad ^{'}:=\frac{d}{dr}. 
\]
The standard form of the above equation can be achieved by using the Liouvillle transformation
\[
v(\rho(r))= \mbox{exp}\left(\frac{1}{2}\int\frac{\rho^{''}(r)}{\rho^{'}(r)}~dr\right)\psi(r) = \sqrt{\rho^{'}(r)}\psi(r),
\]
where we must require that
\begin{equation}\label{condicio}
\rho^{'}(r)>0
\end{equation}
on the interval $(0,\infty)$, i.e. $\rho$ must be an increasing function of the variable $r$. Hence, we end up with the linear ODE
\begin{equation}\label{quasi-ODSE}
\frac{d^2\psi}{dr^2} + J(r)\psi(r) = 0,\quad
J(r) = (\rho^{'})^2 I(\rho(r)) + \frac{1}{2}S(\rho). 
\end{equation}
Here, $S(\rho)$ denotes the Schwarzian derivative of the coordinate transformation $\rho$ evaluated at $r$ and it is given by 
\[
S(\rho) = \frac{\rho^{'''}}{\rho^{'}}-\frac{3}{2}\left(\frac{\rho^{''}}{\rho^{'}}\right)^2.
\]
It can be easily checked that solutions of (\ref{quasi-ODSE})  will be expressed in terms of the solutions of (\ref{TCHeq}) according to  
\[
\psi(r) = \frac{1}{\sqrt{\rho^{'}(r)}}v(\rho(r)).
\]
Moreover, equation (\ref{quasi-ODSE}) will reduce to the radial Schr\"{o}dinger equation (\ref{SH1}) when $J(r) = E - V_{eff}(r)$. Therefore, 
the effective potential is completely determined by the Bose invariant  $I$ \cite{Bose} and the Schwarzian derivative of the coordinate 
transformation. In order to be sure that the potential does not depend on the energy of the particle, we must require that 
\begin{itemize}
\item
the Bose invariant admits a decomposition of the form  
\begin{equation}\label{cc1}
I(\rho)=I_0(\rho)+E I_1(\rho).
\end{equation}
According to Theorem~IV.6 in \cite{Davide2013}, this will be the case if the parameters entering in the triconfluent Heun equation can be written 
as $A_i=a_i+E b_i$ with $i=0,1,2$ and $a_i,b_i\in\mathbb{R}$ for any $i=0,1,2$. Then, we have
\[
I_0(\rho)=a_0+a_1\rho+a_2\rho^2-\frac{9}{4}\rho^4,\quad
I_1(\rho)=b_0+b_1\rho+b_2\rho^2. 
\]
\item
The coordinate transformation satisfies the differential equation 
\begin{equation}\label{cc2}
\left(\rho^{'}\right)^2 I_1(\rho(r))=1.
\end{equation}
The condition expressed by (\ref{condicio}) implies that we have to take the positive square root of (\ref{cc2}), that is
\begin{equation}\label{cc3}
\rho^{'}(r)=\frac{1}{\sqrt{I_1(\rho(r))}}. 
\end{equation}
\end{itemize}
With the help of (\ref{cc1}) and (\ref{cc2}) the effective potential can be written as
\[
V_{eff}(r)=-\frac{I_0(\rho(r))}{I_1(\rho(r))}-\frac{1}{2}S(\rho).
\]
Rewriting the Schwarzian derivative in terms of $I_1$ and its derivatives with respect to $\rho$ we obtain the most general form for the  
effective potential such that the radial Schr\"{o}dinger equation (\ref{SH1}) can be transformed into a triconfluent Heun equation, namely 
\begin{equation}\label{6}
V_{eff}(r) = -\frac{I_0(\rho(r))}{I_1(\rho(r))}+\frac{4I_1(\rho(r))\ddot{I_1}(\rho(r))-5\dot{I_1}^2(\rho(r))}{16I_1^3(\rho(r))},\quad 
\dot{}:=\frac{d}{d\rho}.
\end{equation} 
Replacing the corresponding expressions for $I_0$ and $I_1$ into (\ref{6}) we can write the effective potential as
\begin{equation}\label{8}
V_{eff}(r) = -\frac{4a_0+4a_1\rho+4a_2\rho^2-9\rho^4}{4(b_0+b_1\rho+b_2\rho^2)}
-\frac{12b_2^2\rho^2+12b_1b_2\rho+5b_1^2-8b_0 b_2}{16(b_0+b_1\rho+b_2\rho^2)^3},\quad \rho=\rho(r).
\end{equation}
Note that (\ref{8}) represents a class of potentials depending upon the six real parameters $a_i,b_i$ with $i=0,1,2$. Concerning the solution of 
the differential equation governing the coordinate transformation we need to require that $I_1(\rho(r))>0$ for $r\in(0,\infty)$. 
Let $\Delta=b_1^2-4b_0 b_2$. We have the following cases
\begin{enumerate}
 \item 
$b_0,b_1,b_2\neq 0$ and $b_2>0$. 
\begin{enumerate}
\item
If $\Delta<0$, the coordinate transformation $\rho$ maps the interval $(0,\infty)$ of the radial variable into the whole real line. Using $2.261$ 
and $2.262(1)$ in \cite{Gradsh} we find that the solution of (\ref{cc3}) can be written as
\[
r(\rho)=\frac{1}{2}\left(\rho+\frac{b_1}{2b_2}\right)\sqrt{I_1}-\frac{\Delta}{8b_2\sqrt{b_2}}
\mbox{arcsinh}{\left(\frac{2b_2\rho+b_1}{\sqrt{-\Delta}}\right)}. 
\]
As $\rho\to\infty$ we see that $r\approx(\sqrt{b_2}/2)\rho^2$.
\item
Let $\Delta=0$. Since $b_1^2>0$, $b_2>0$, and $b_1^2=4b_0 b_2$, then $b_0>0$ and $b_1=\pm 2\sqrt{b_0 b_2}$. Hence, 
$I_1=(\sqrt{b_2}\rho\pm\sqrt{b_0})^2$ and equation (\ref{cc3}) becomes
\[
\rho^{'}(r)=\frac{1}{\sqrt{b_2}\rho\pm\sqrt{b_0}},
\]
where the plus sign must be taken when $b_1>0$. Note that $\rho^{'}>0$ whenever $\rho>\mp\sqrt{b_0/b_2}$ where the minus sign must be taken 
for $b_1>0$. Integrating the above equation we obtain
\[
r(\rho)=\frac{\sqrt{b_2}}{2}\rho^2\pm\sqrt{b_0}\rho. 
\]
The above equation is quadratic in $\rho$ and therefore we can solve it in order to express $\rho$ as a function of the radial variable. In the 
case $b_1>0$ and requiring that $\rho$ is an increasing function of the radial variable we find that
\[
\rho(r)=\frac{-\sqrt{b_0}+\sqrt{b_0+2\sqrt{b_2}r}}{\sqrt{b_2}}. 
\]
In this case $\rho$ maps the interval $(0,\infty)$ into itself and the condition $\rho>-\sqrt{b_0/b_2}$ is automatically satisfied. If $b_1<0$, by 
a similar reasoning we find the solution
\[
\rho(r)=\frac{\sqrt{b_0}+\sqrt{b_0+2\sqrt{b_2}r}}{\sqrt{b_2}}. 
\]
The image of the interval $(0,\infty)$ under the transformation $\rho$ is the interval $(2\sqrt{b_0/b_2},\infty)$ and also in this case the condition 
$\rho>\sqrt{b_0/b_2}$ is satisfied. If $b_0=0$, then $\Delta=0$ implies $b_1=0$ and $\rho$ is given by the expression
\[
\rho(r)=\sqrt{\frac{2r}{\sqrt{b_2}}}. 
\]
\item
If $\Delta>0$, then $I_1>0$ on the interval $(-\infty,\rho_1)\cup(\rho_2,\infty)$ where $\rho_1$ and $\rho_2$ denote the roots of $I_1$. In this 
case we can express the radial variable in terms of $\rho$ as
\[
r(\rho)=\frac{1}{2}\left(\rho+\frac{b_1}{2b_2}\right)\sqrt{I_1}-\frac{\Delta}{8b_2\sqrt{b_2}}
\ln{\left(2\sqrt{b_2 I_1}+2b_2\rho+b_1\right)}. 
\]
\end{enumerate}
\item
$b_0,b_1,b_2\neq 0$ and $b_2<0$. In this case only if $\Delta>0$ we can make $I_1>0$ on an open interval $(\rho_1,\rho_2)$ where $\rho_1$ and 
$\rho_2$ denote the roots of $I_1$. Employing $2.261$ and $2.262(1)$ in \cite{Gradsh} yield
\[
r(\rho)=\frac{1}{2}\left(\rho+\frac{b_1}{2b_2}\right)\sqrt{I_1}+\frac{\Delta}{8b_2\sqrt{-b_2}}
\mbox{arcsin}\left(\frac{2b_2\rho+b_1}{\sqrt{\Delta}}\right). 
\]
\item
If $b_2=0$ and $b_1\neq 0$, then $I_1(\rho)=b_1\rho+b_0$ and we must require that $\rho>-b_0/b_1$. The solution of the ODE governing the coordinate 
transformation is
\[
r(\rho)=\frac{2}{3b_1}(b_1\rho+b_0)^{3/2}. 
\]
Moreover, we can express $\rho$ in terms of the radial coordinate as
\[
\rho(r)=-\frac{b_0}{b_1}+\frac{1}{b_1}\left(\frac{3b_1 r}{2}\right)^{2/3}. 
\]
We can observe that $\rho$ is an increasing function of $r$ and it maps the interval $(0,\infty)$ to $(-b_0/b_1,\infty)$. If $b_0=0$ the above 
expression simplifies to
\[
\rho(r)=\left(\frac{3}{2}\right)^{2/3}b_1^{-1/3}r^{2/3}. 
\]
\end{enumerate}
\section{Analysis of the potentials}
We analyze those potentials arising from the cases when the coordinate transformation $\rho$ can be written as an explicit function 
of the real variable $r$.
\subsection{The case $\Delta=0$, and $b_0,b_1,b_2>0$}
It is straightforward to verify that $I_1=b_0+2\sqrt{b_2}r$ and the effective potential can be 
written as
\begin{equation}\label{vi1}
V_{1}(r)=c_0+\frac{c_1}{\sqrt{b_0+2\sqrt{b_2}r}}+\frac{c_2}{b_0+2\sqrt{b_2}r}+\frac{c_3}{(b_0+2\sqrt{b_2}r)^2}
+c_4\sqrt{b_0+2\sqrt{b_2}r}+c_5r
\end{equation}
with
\[
c_0=\frac{63 b_0}{4b_2^2}-\frac{a_2}{b_2},\quad
c_1=\frac{2a_2\sqrt{b_0}}{b_2}-\frac{a_1}{\sqrt{b_2}}-\frac{9b_0\sqrt{b_0}}{b_2^2},\quad
c_2=-a_0+\frac{a_1\sqrt{b_0}}{\sqrt{b_2}}+\frac{9b_0^2}{4b_2^2}-\frac{a_2 b_0}{b_2},
\]
\[
c_3=-\frac{3}{4}b_2,\quad c_4=-\frac{9\sqrt{b_0}}{b_2^2},\quad c_5=\frac{9}{2b_2\sqrt{b_2}} 
\]
Note that $c_0=c_1=c_2=0$ if we choose
\[
a_0=9\left(\frac{b_0}{b_2}\right)^2,\quad
a_1=\frac{45}{2}\left(\frac{b_0}{b_2}\right)^{3/2},\quad
a_2=\frac{63b_0}{4b_2}.
\]
and the corresponding effective potential will depend only on the parameters $b_0$ and $b_2$. We observe that
\begin{enumerate}
\item
Since $r>0$ and $b_0>0$, it follows that the denominators appearing in (\ref{vi1}) can never vanish and therefore the potential is never singular 
on the open interval $(0,\infty)$.
\item
If $r=0$ the potential is finite there while asymptotically for $r\to\infty$ it grows linearly as
\[
V_{1}(r)\approx\frac{9r}{2b_2\sqrt{b_2}}.
\]
\item
There are choices of the parameters such that $V_1$ has at least a global minimum on the positive real line. See Fig.~\ref{fig1}.
\begin{figure}\label{F1}
\includegraphics[scale=0.5]{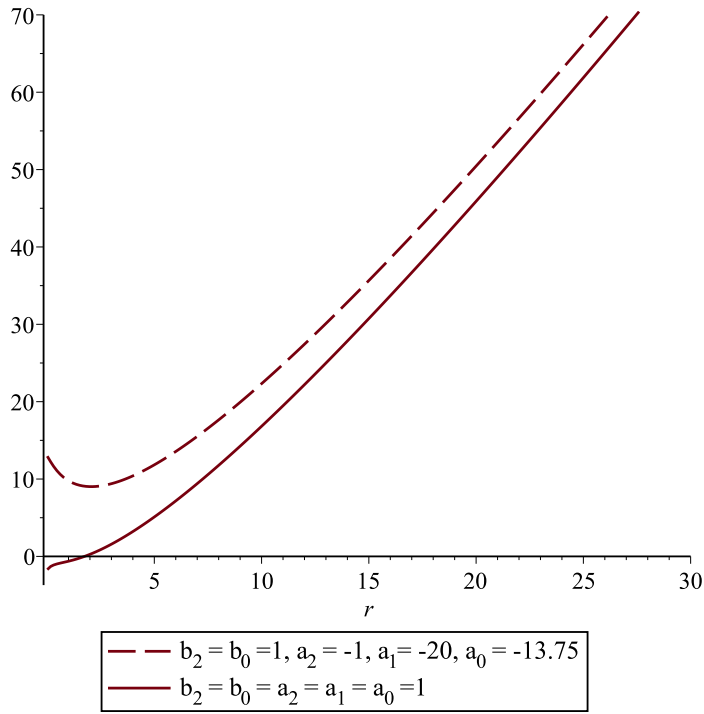}
\caption{\label{fig1}
Behaviour of the potential (\ref{vi1}) for different choices of the parameters.
}
\end{figure}
\end{enumerate}
The study of the critical points of this potential leads to a polynomial equation of degree eight. 
Applying Descartes' rule of signs we find that the potential (\ref{vi1}) has four, two, or no positive critical points whenever $c_1$ and $c_3$ are negative and $c_2$ is positive; three or one 
positive critical point if $c_1<0$ and $c_2,c_3>0$, and only one positive critical point if $c_1$, 
$c_2$, and $c_3$ are positive. For other choices of the signs of the coefficients $c_1$, 
$c_2$, and $c_3$ there are two or no positive critical point.

\subsection{The case $\Delta=0$, $b_0,b_2>0$, and $b_1<0$}
The function $I_1$ is formally given as in the previous case and the effective potential has the same form as (\ref{vi1}) with coefficients 
$d_0,\cdots,d_5$ given by
\[
d_0=c_0,\quad
d_1=-\frac{2a_2\sqrt{b_0}}{b_2}-\frac{a_1}{\sqrt{b_2}}+\frac{9b_0\sqrt{b_0}}{b_2^2},\quad
d_2=a_0-\frac{a_1\sqrt{b_0}}{\sqrt{b_2}}+\frac{9b_0^2}{4b_2^2}-\frac{a_2 b_0}{b_2},\quad
d_3=c_3,\quad d_4=-c_4,\quad d_5=c_5. 
\]
Note that $d_0=d_1=d_2=0$ if we choose
\[
a_0=-36\left(\frac{b_0}{b_2}\right)^2,\quad
a_1=-\frac{45}{2}\left(\frac{b_0}{b_2}\right)^{3/2},\quad
a_2=\frac{63b_0}{4b_2}.
\]
We observe that
\begin{enumerate}
\item
Since $r>0$ and $b_0>0$, it follows that $I_1$ can never vanish and therefore the potential is never singular.
\item
The potential is finite at $r=0$ while asymptotically for $r\to\infty$ it grows linearly as (\ref{vi1}). 
\item
There are again choices of the parameters such that the potential has at least one minimum.
\end{enumerate}

\subsection{The case $\Delta=0$, $b_2>0$, and $b_0=b_1=0$}
The potential takes the form
\begin{equation}\label{v3}
V_3(r)=e_0+\frac{e_1}{\sqrt{r}}+\frac{e_2}{r}+\frac{e_3}{r^2}+e_4 r
\end{equation}
where
\[
e_0=-\frac{a_2}{b_2},\quad
e_1=-\frac{a_1}{\sqrt{2} b_2^{3/4}},\quad
e_2=-\frac{a_0}{2\sqrt{b_2}},\quad
e_3=-\frac{3}{16},\quad
e_4=c_5.
\]
Note that it becomes singular at $r=0$ and asymptotically at infinity it behaves as (\ref{vi1}). There are choices of the parameters such that this potential has a maximum very close to $r=0$ and a minimum. See Fig.~\ref{fig11}. 
\begin{figure}\label{F11}
\includegraphics[scale=0.5]{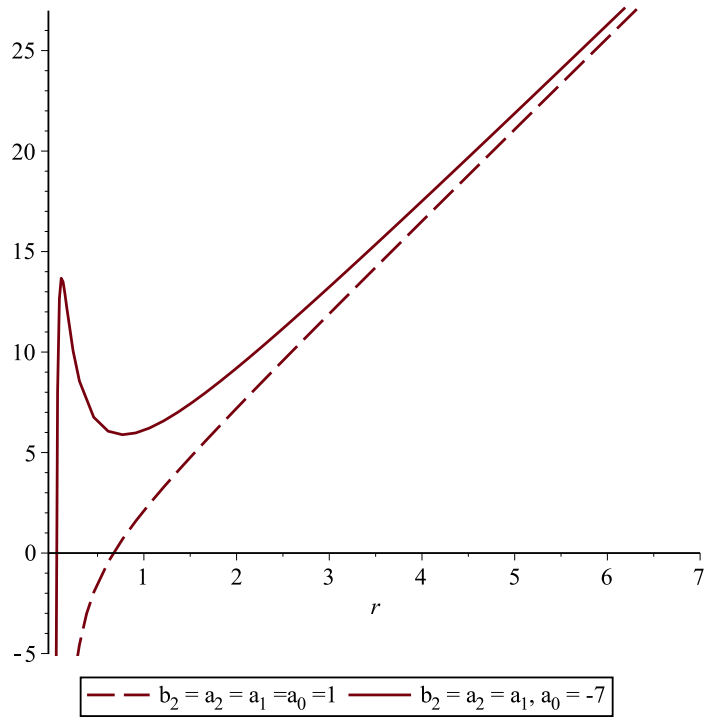}
\caption{\label{fig11}
Behaviour of the potential for the case $\Delta=0$, $b_0,b_2>0$, and $b_1<0$.
}
\end{figure}
The potential (\ref{v3}) is singular at $r=0$ and it grows linearly for $r\to\infty$. We have choices of the parameters such that the potential exhibits maxima and minima. If we make the substitution $r=\tau^2$, the critical points of the potential (\ref{v3}) must satisfy the sextic equation 
\begin{equation}\label{sextic1}
2e_4\tau^6-e_1\tau^3-2e_2\tau^2-4e_3=0.
\end{equation}
If $a_0$ and $a_1$ are both positive, there is no change of sign in (\ref{sextic1}) and Descartes' rule of signs implies that the sextic equation will not have any positive real root, and hence the potential has no extrema. For all other choices of the signs of the coefficients $a_0$, and $a_1$, there are only two sign changes in (\ref{sextic1}) 
and hence the potential will admit two positive critical points or none. For a complete root classification of a sextic equation we refer to \cite{Yang}. 

\subsection{The case $b_2=0$}
It can be easily checked that $I_1(r)=(3b_1 r/2)^{2/3}$ and the potential can be written as
\begin{equation}\label{int}
V_4(r)=v_0+\frac{v_1}{r^{2/3}}+\frac{v_2}{r^2}+v_3r^{2/3}+v_4r^{4/3}+v_5r^2
\end{equation}
with
\[
v_0=-\frac{a_1}{b_1}+\frac{2a_2 b_0}{b_1^2}-\frac{9b_0^3}{b_1^4},\quad
v_1=\left(\frac{2}{3b_1}\right)^{2/3}\left(\frac{a_1 b_0}{b_1}-a_0-\frac{a_2 b_0^2}{b_1^2}+\frac{9b_0^4}{4b_1^4}\right),\quad
v_2=-\frac{5}{36},
\]
\[
v_3=\frac{1}{b_1^{2/3}}\left(\frac{2}{3b_1}\right)^{2/3}\left[\frac{27}{2}\left(\frac{b_0}{b_1}\right)^2-a_2\right],\quad
v_4=-\frac{3^{10/3}b_0}{2^{4/3}b_1^{8/3}},\quad
v_5=\frac{81}{16 b_1^2}.
\]
We can set $v_0=v_1=v_3=0$ if we choose
\[
a_0=\frac{27}{4}\left(\frac{b_0}{b_1}\right)^4,\quad
a_1=18\left(\frac{b_0}{b_1}\right)^3,\quad
a_2=\frac{27}{2}\left(\frac{b_0}{b_1}\right)^2.
\]
The potential (\ref{int}) is singular at $r=0$ and it grows quadratically for $r\to\infty$. We have choices of the parameters such that the potential exhibits maxima and minima. If we make the substitution $r=\tau^{3/2}$, it is not difficult to verify that the critical points of the potential (\ref{int}) must satisfy the sextic equation 
\begin{equation}\label{sextic}
3v_5\tau^6+2v_4\tau^5+v_3\tau^4-v_1\tau^2-3v_2=0.
\end{equation}
If $v_1<0$ and $v_3,v_4>0$, we do not have any change of sign in (\ref{sextic}) and Descartes' rule of signs implies that the sextic equation will not have any positive real root, and therefore the potential has no extrema. If $v_1,v_2>0$ and $v_4<0$, there are four changes of sign and the potential may have four, two or no positive critical point. For all other choices of the signs of the coefficients $v_1$, $v_3$, and $v_4$ there are only two sign changes in (\ref{sextic}) 
and hence the potential will admit two positive critical points or none. In the subcase $b_2=b_0=0$ the potential (\ref{int}) simplifies to
\begin{equation}\label{V5}
V_5(r)=u_0+\frac{u_1}{r^{2/3}}+\frac{u_2}{r^2}+u_3r^{2/3}+u_4r^2
\end{equation}
where
\[
u_0=-\frac{a_1}{b_1},\quad u_1=-a_0\sqrt[3]{\frac{4}{9b_1^2}},\quad
u_2=v_2,\quad
u_3=-a_2\sqrt[3]{\frac{4}{9b_1^4}},\quad
u_4=\frac{81}{16 b_1^2}.
\]
\begin{figure}\label{F4}
\includegraphics[scale=0.5]{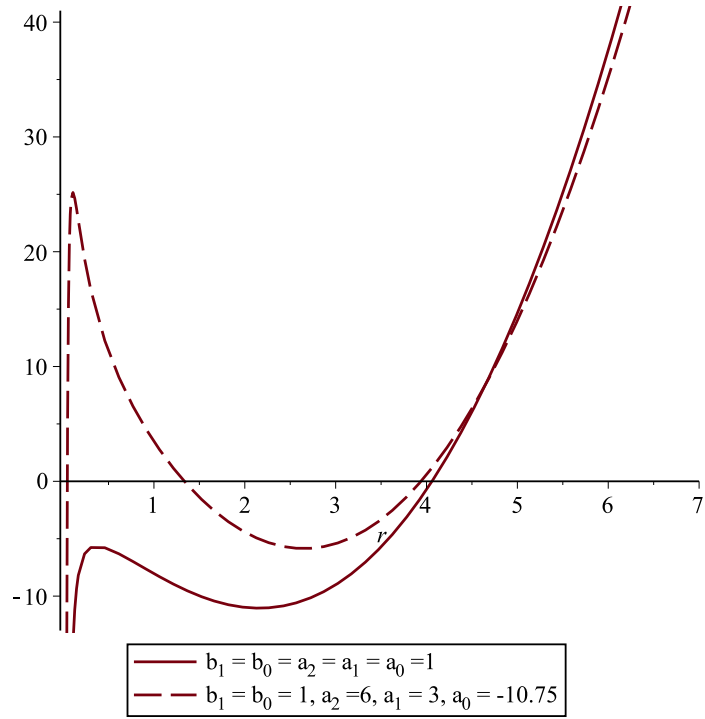}
\caption{\label{fig4}
Behaviour of the potential (\ref{int}) for different values of the parameters.
}
\end{figure}
If we make the substitution $r=\tau^{3/4}$, it is not difficult to verify that the critical points of the potential (\ref{V5}) must satisfy the cubic equation 
\begin{equation}\label{cubic}
3u_4\tau^3+u_3\tau^2-u_1\tau-3u_2=0.
\end{equation}
If $a_2<0$ and $a_0>0$, we do not have any change of sign in (\ref{cubic}) and Descartes' rule of signs implies that the cubic equation will not have any positive real root. Hence, the potential has no positive extrema. For any other choice of the signs of the parameters $a_2$ and $a_1$, we will always have two sign changes in (\ref{cubic}) and therefore the potential will admit two positive critical points or none.

\section{Solutions of the radial Schr\"{o}dinger equation}
We study the behaviour of the solutions of the radial Schr\"{o}dinger equation for the general class of potentials (\ref{8}) and for the potentials 
discussed in the previous section. First of all, by means of (\ref{cc3}) the solution of the radial Schr\"{o}dinger equation can be written as
\[
\psi(r)=\sqrt[4]{b_0+b_1\rho+b_2\rho^2}v(\rho),
\]
where $\rho=\rho(r)$ and $v$ is a solution of the triconfluent Heun equation (\ref{TCHeq}). It is in general convenient to bring (\ref{TCHeq}) into 
its canonical form \cite{Ronveaux} with the help of the substitution
\[
v(\rho)=e^{\frac{A_2}{3}\rho-\frac{\rho^3}{2}}y(\rho), 
\]
where the unknown function $y$ must satisfy the equation
\begin{equation}\label{puntoh}
\frac{d^2 y}{d\rho^2}-(\gamma+3\rho^2)\frac{dy}{d\rho}+[\alpha+(\beta-3)\rho]=0 
\end{equation}
with
\begin{equation}\label{abc}
\alpha=A_0+\frac{A_2^2}{9},\quad \beta=A_1,\quad \gamma=-\frac{2}{3}A_2. 
\end{equation}
Hence, the general solution of the radial Schr\"{o}dinger equation can be cast in the form
\begin{equation}\label{stellah}
\psi(r)=\sqrt[4]{b_0+b_1\rho+b_2\rho^2}e^{\frac{A_2}{3}\rho-\frac{\rho^3}{2}}y(\rho).
\end{equation}
Note that the function multiplying $y$ in (\ref{stellah}) decays exponentially in $r$ as $r\to+\infty$ because in the case $b_0,b_1,b_2\neq 0$ 
we have
\[
\rho(r)\approx 
\left\{ \begin{array}{ll}
         \sqrt{\frac{2r}{\sqrt{b_2}}} & \mbox{if $b_2 > 0$}\\
        \rho_2+\frac{1}{\sqrt[3]{|b_2|(\rho_2-\rho_1)}}\left(\frac{3r}{2}\right)^{2/3} & \mbox{if $b_2 < 0$}.\end{array}\right.
\]
with $\rho_1$ and $\rho_2$ denoting the roots of the quadratic polynomial $I_1$ introduced in Section~\ref{red} whereas if $b_2=0$ and $b_1\neq 0$ 
we get
\[
\rho(r)\approx\frac{1}{b_1}\left(\frac{3b_1 r}{2}\right)^{2/3}. 
\]
Since equation (\ref{TCHeq}) has no finite singular points, we can construct Taylor series for the solutions of the triconfluent Heun equation by 
adopting the method outlined in \cite{Ronveaux}. We first rewrite (\ref{puntoh}) in terms of the Euler operator $\delta=\rho d/d\rho$ as
\[
\delta^2 y-(1+\gamma\rho+3\rho^3)\delta y+(\beta-3)\rho^3 y=0. 
\]
Observing that $\rho^n$ with $n\in\mathbb{N}$ are eigenfunctions of the Euler operator, one finds the following two linearly independent solutions, 
namely
\begin{equation}\label{T1T2}
T_1(\alpha,\beta,\gamma;\rho)=\sum_{n=0}^\infty e_n(\alpha,\beta,\gamma)\rho^n,\quad
T_2(\alpha,\beta,\gamma;\rho)=\sum_{n=0}^\infty s_n(\alpha,\beta,\gamma)\rho^{n+1} 
\end{equation}
with
\[
e_0(\alpha,\beta,\gamma)=1,\quad e_1(\alpha,\beta,\gamma)=0,\quad e_2(\alpha,\beta,\gamma)=-\frac{\alpha}{2},
\]
\[
e_n(\alpha,\beta,\gamma)=\frac{(n-1)\gamma e_{n-1}(\alpha,\beta,\gamma)-\alpha e_{n-2}(\alpha,\beta,\gamma)
-(\beta+6-3n)e_{n-3}(\alpha,\beta,\gamma)}{n(n-1)},\quad n\geq 3
\]
and
\[
s_0(\alpha,\beta,\gamma)=1,\quad s_1(\alpha,\beta,\gamma)=\frac{\gamma}{2},\quad s_2(\alpha,\beta,\gamma)=\frac{\gamma^2-\alpha}{6},
\]
\[
s_n(\alpha,\beta,\gamma)=\frac{n\gamma s_{n-1}(\alpha,\beta,\gamma)-\alpha s_{n-2}(\alpha,\beta,\gamma)
-(\beta+3-3n)s_{n-3}(\alpha,\beta,\gamma)}{n(n+1)},\quad n\geq 3.
\]
We will refer to (\ref{T1T2}) as the general solution to contrast it with 
the polynomial one which will be discussed below. Worth mentioning is also
the convergence of (\ref{T1T2}) as shown in \cite{Ronveaux}.

In order to find polynomial solutions of (\ref{puntoh}) we start AGAIN by assuming a power solution of the form
\begin{equation}\label{dpuntih}
y(\rho)=\sum_{n=0}^\infty w_n\rho^n. 
\end{equation}
Substituting (\ref{dpuntih}) into (\ref{puntoh}) we obtain the following recurrence relation for the coefficients $w_n$
\begin{equation}\label{ric}
(\beta-3n)w_{n-1}+\alpha w_n-\gamma(n+1)w_{n+1}+(n+1)(n+2)w_{n+2}=0 
\end{equation}
which holds for all $n\geq 0$ provided that $w_{-1}=0$.
The above recursion relation is the same as the recursion relation for $e_n$ (see above) with a shift for $n$.
In order to have a polynomial solution of degree $N$ we require that $w_n=0$ for all 
$n>N$. Moreover, for $n=N+1$ the recurrence relation above gives the condition $\beta=3(N+1)$. Taking into account that $\beta=a_1+b_1 E$, such a 
condition gives the energy eigenvalue
\[
E_N=\frac{3(N+1)-a_1}{b_1},\quad b_1\neq 0
\]
with $N=0,1,2,\cdots$. Furthermore, for $0\leq n\leq N$ the recursion formula (\ref{ric}) generates a system of $N+1$ homogenous linear equations for the coefficients 
$w_0,w_1,\cdots,w_N$. This system admits a non-trivial solution if and only if the determinant of the $(N\times 1)\times(N\times 1)$ matrix 
\[
D_{N+1}= \begin{pmatrix}
  \alpha     &     -\gamma    &    2     &     0    &    0     &    \cdots   &         &         &              &         0  \\
  3N         &      \alpha    & -2\gamma & 2\cdot 3 &    0     &    \cdots   &         &         &              &            \\
   0         &     3(N-1)     &  \alpha  & -3\gamma & 3\cdot 4 &    \cdots   &         &         &              &            \\
\vdots       & \vdots         &\vdots    &\vdots    &\vdots    &    \ddots   &         &         &              &     \vdots \\
             &                &          &          &          &             &3\cdot 3 & \alpha  & -(N-1)\gamma &    N(N-1)  \\
             &                &          &          &          &             & 0       &3\cdot 2 & \alpha       & -N\gamma   \\
   0         &                &          &          &          &   \cdots    & 0       &0        &3\cdot 1      & \alpha
 \end{pmatrix}
\]
does not vanish. The first polynomials are given by 
\begin{itemize}
 \item 
\underline{$N=0$}: $\beta=3$, $\alpha=0$ and $p_0(\rho)=1$;
\item
\underline{$N=1$}: $\beta=6$, $\alpha^2+3\gamma=0$ and $p_1(\rho)=\rho-\alpha/3$;
\item
\underline{$N=2$}: $\beta=9$, $\alpha^3+12\alpha\gamma+36=0$ and
\[
p_2(\rho)=\rho^2-\frac{\alpha}{3}\rho+\frac{\alpha^2}{36}-\frac{1}{\alpha}; 
\]
\item
\underline{$N=3$}: $\beta=12$, $\alpha^4+30\alpha^2\gamma+216\alpha+81\gamma^2=0$ and
\[
p_3(\rho)=\rho^3-\frac{\alpha}{3}\rho^2+\left(\frac{\gamma}{2}+\frac{\alpha^2}{18}\right)\rho-\frac{\alpha^3}{162}-\frac{7}{54}\alpha\gamma-\frac{2}{3}; 
\]
\item
\underline{$N=4$}: $\beta=15$, $\alpha^5+60\alpha^3\gamma+756\alpha^2+576\alpha\gamma^2+5184\gamma=0$ and
\[
p_4(\rho)=\rho^4-\frac{\alpha}{3}\rho^3+\left(\frac{\alpha^2}{18}+\frac{2}{3}\gamma\right)\rho^2-\left(\frac{\alpha^3}{162}+\frac{5}{27}\alpha\gamma
+\frac{4}{3}\right)\rho+\frac{\alpha^4}{1944}+\frac{2}{81}\alpha^2\gamma+\frac{5}{18}\alpha+\frac{\gamma^2}{9}. 
\]
\end{itemize}
Note that the above list extends the one given in \cite{Ronveaux}. The special case $b_1=0=b_2$ and $b_0=1$ and in particular the boundary value problem
\begin{equation}\label{ep}
-\frac{d^2 v}{d\rho^2}+P(\rho)v=Ev,\quad v(-\infty)=v(+\infty)=0,\quad P(\rho)=\frac{9}{4}\rho^4-a_2\rho^2-a_1\rho-a_0
\end{equation}
has been considered by \cite{Eremenko1}. Note that the boundary condition is equivalent to the requirement $v\in L^2(\mathbb{R})$. According to \cite{Berezin,Sibuya} the spectrum of the above problem is discrete, all eigenvalues are real and simple and they can be arranged into an increasing sequence $E_0<E_1<\cdots$. Moreover, for $n\to\infty$ the asymptotic behaviour of the eigenvalues is \cite{Eremenko1}
\[
E_n\approx
\left[\frac{4\pi n}{2B(3/2,1/4)}\right]^{4/3}=\left[\frac{3\Gamma^2(3/4)}{\sqrt{2\pi}} n\right]^{4/3},
\]
where $B$ denotes the Euler function. Furthermore,  all non-real zeros of the eigenfunctions of the problem (\ref{ep}) belong to the imaginary axis (see Theorem~1 in \cite{Eremenko2}). 

We conclude this section by deriving a formula for the solution of the recurrence relation (\ref{ric}). By shifting the index according 
to the prescription $n\to n+1$ we can rewrite (\ref{ric}) as a third order homogeneous linear difference equation
\begin{equation}\label{cuadro}
w_{n+3}=-\pi_1(n,\gamma) w_{n+2}-\pi_2(n,\alpha)w_{n+1}-\pi_3(n,\beta)w_n
\end{equation}
with $n\geq -2$ and
\[
\pi_1(n,\gamma)=-\frac{\gamma}{n+3},\quad\pi_2(n,\alpha)=\frac{\alpha}{(n+2)(n+3)},\quad \pi_3(n,\beta)=\frac{\beta-3(n+1)}{(n+2)(n+3)},\quad
w_{-2}=w_{-1}=0.
\]
By letting $n=0$ we obtain $w_3$ in terms of $w_2,w_1$, and $w_0$, namely
\[
w_3=\frac{\gamma}{3}w_2-\frac{\alpha}{2\cdot 3}w_1-\frac{\beta-3}{2\cdot 3}w_0.
\]
Furthermore, for $n=-2$ and $n=-1$ we get, respectively
\[
w_1=\gamma w_0,\quad w_2=\frac{1}{2}(\gamma^2-\alpha)w_0
\]
from which it follows that
\[
w_3=\frac{1}{3!}\left[\gamma^3-2\alpha\gamma-(\beta-3)\right]w_0,\quad
w_4=\frac{1}{4!}\left[\gamma^4-3\alpha\gamma^2+\alpha^2-3\gamma(\beta-5)\right]w_0.
\]
Hence, $w_n$ can be evaluated once  $w_0$ is assigned. Note also that solving (\ref{cuadro}) under the requirement $w_{-2}=w_{-1}=0$ is equivalent to solve the same recurrence 
relation subject to the initial conditions
\begin{equation}\label{tildazza}
w_0=1,\quad w_1=\gamma,\quad w_2=\frac{\gamma^2-\alpha}{2}.
\end{equation}
Since our difference equation is linear and of third order, it will have a unique solution which can be expressed as a linear combination of three linear independent solutions 
$w_n^{(1)}$, $w_n^{(2)}$, and $w_n^{(3)}$. Once we find by other means these three particular solutions, we can check their linear independence by verifying that the Casoratian 
$W_n$ of the solutions $w_n^{(1)}$, $w_n^{(2)}$, and $w_n^{(3)}$ defined by
\[
W_n=\mbox{det}
\left(
\begin{array}{ccc}
w_n^{(1)} & w_n^{(2)}& w_n^{(3)}\\
w_{n+1}^{(1)} & w_{n+1}^{(1)}& w_{n+1}^{(1)}\\
w_{n+2}^{(1)} & w_{n+2}^{(1)}& w_{n+2}^{(1)}
\end{array}
\right)
\]
does not vanish for any value of $n$. On the other hand, Abel's lemma allows us to compute the Casoratian even without knowing explicitly the solutions $w_n^{(i)}$ with 
$i=1,2,3$. Applying $(2.2.10)$ in \cite{Elaydi} we find
\[
W_n=(-)^{3n}\prod_{k=0}^{n-1}\pi_3(k,\beta)W_0=\frac{2(-)^{3n}(n+2)}{(n+2)!^2}(\beta-3)(\beta-6)\cdots(\beta-3n)W_0.
\]
It is clear that we will have three linearly independent solutions provided that $\beta\neq 3(k+1)$ for all $k\in\mathbb{N}$ and $W_0\neq 0$. If this is the case, then the general solution of (\ref{cuadro}) can be written as
\[
w_n=C_1 w_n^{(1)}+C_2 w_n^{(2)}+C_3 w_n^{(3)}
\]
with constants $C_1$, $C_2$, and $C_3$ determined by the initial condition (\ref{tildazza}). In order to solve the recurrence relation (\ref{cuadro}) we first rewrite it as a system 
of first order equations of dimension $3$, namely
\begin{equation}\label{stelladav}
Z_{n+1}=A(\alpha,\beta,\gamma;n)Z_n
\end{equation}
with
\[
Z_n=\left(
\begin{array}{c}
Z^{(1)}_n\\
Z^{(2)}_n\\
Z^{(3)}_n
\end{array}
\right)=\left(
\begin{array}{c}
w_n\\
w_{n+1}\\
w_{n+2}
\end{array}
\right),\quad
A(\alpha,\beta,\gamma;n)=\left(
\begin{array}{ccc}
0&1&0\\
0&0&1\\
-\pi_3(n,\beta) & -\pi_2(n,\alpha) & -\pi_1(n,\gamma)
\end{array}
\right),
\]
where $A$ is the so-called companion matrix of the recurrence relation (\ref{cuadro}). Clearly, this matrix is non-singular if its determinant does not vanish and this condition translates into the requirement $\beta\neq 3(n+1)$ for any $n\in\mathbb{N}$. We want to solve (\ref{stelladav}) subject to the initial condition
\[
Z_0=\left(
\begin{array}{c}
w_1\\
w_{1}\\
w_{2}
\end{array}
\right)=\left(
\begin{array}{c}
1\\
\gamma\\
\frac{\gamma^2-\alpha}{2}
\end{array}
\right).
\]
By Theorem~$3.4$ in \cite{Elaydi} this initial value problem has a unique solution $U_{n}$ such that $U_0=Z_0$. From (\ref{stelladav}) we have
\[
U_1=A(\alpha,\beta,\gamma;0)Z_0,\quad U_2=A(\alpha,\beta,\gamma;1)Z_1=A(\alpha,\beta,\gamma;1)A(\alpha,\beta,\gamma;0)Z_0
\]
and by induction we conclude that
\[
U_n=\left(\prod_{k=0}^{n-1}A(\alpha,\beta,\gamma;k)\right)Z_0,
\]
where
\[
\prod_{k=0}^{n-1}A(\alpha,\beta,\gamma;k)=\left\{ \begin{array}{ll}
        A(\alpha,\beta,\gamma;n-1)A(\alpha,\beta,\gamma;n-2)\cdots A(\alpha,\beta,\gamma;0) & \mbox{if $n \geq 0$}\\
        \mathbb{I}_3 & \mbox{if $n= 0$}.\end{array} \right. .
\]
Even though the above formula contains the solution of the recurrence relation (\ref{cuadro}), it is difficult to extract from it the behaviour of the solution $w_n$ as $n\to\infty$. The next result circumvents this problem by investigating the asymptotic behaviour of $w_n$ with the help of the so-called Z-transform method which reduces the study of a linear 
difference equation to an examination of a corresponding complex function. 
\begin{theorem}\label{TM}
For the recurrence relation (\ref{cuadro}) with the initial condition (\ref{tildazza}) it holds
\[
\lim_{n\to\infty}w_n=0.
\]
\end{theorem}
\begin{proof}
Applying the definition of the Z-transform to our recurrence relation (\ref{cuadro}), i.e.
\[
\widetilde{x}(z)=Z(w_n)=\sum_{j=0}^\infty\frac{w_j}{z^j},\quad z\in\mathbb{C}\backslash\{0\},
\]
we find with the help of the properties of the Z-transform listed in Ch.~$6$ in \cite{Elaydi} that the original recurrence relation gives rise to the following second order, linear, non homogeneous complex differential equation
\begin{equation}\label{equazza}
z^5\frac{d^2\widetilde{x}}{dz^2}+z(2z^3+\gamma z^2+3)\frac{d\widetilde{x}}{dz}+(\alpha z+\beta-3)\widetilde{x}=(2w_2-\gamma w_1+\alpha w_0)z.
\end{equation}
The initial condition (\ref{tildazza}) implies that $2w_2-\gamma w_1+\alpha w_0=0$ and therefore, we are left with the following homogeneous differential equation
\begin{equation}\label{ascia}
z^5\frac{d^2\widetilde{x}}{dz^2}+z(2z^3+\gamma z^2+3)\frac{d\widetilde{x}}{dz}+(\alpha z+\beta-3)\widetilde{x}=0.
\end{equation}
By means of the substitution $\omega=1/z$ the above equation can be transformed into the triconfluent Heun equation (\ref{puntoh}). Hence, a particular solution is given by 
$\widetilde{x}_1(\omega)=T_1(\alpha,\beta,\gamma;\omega)$ with $T_1$ defined as in (\ref{T1T2}). In order to construct a linearly independent set we make use of Preposition~
$4.2$ with $j=1$ in \cite{Decarreau}. This gives
\[
\gamma T_2(\alpha,\beta,\gamma;\omega)+T_1(\alpha,\beta,\gamma;\omega)=e^{\omega^3+\gamma\omega}T_1(\alpha,-\beta,\gamma;-\omega)
\]
with $T_2$ defined as in (\ref{T1T2}). Hence, the general solution of (\ref{ascia}) for $\gamma\neq 0$ is
\[
\widetilde{x}(z)=C_1 T_1(\alpha,\beta,\gamma;1/z)+C_2 e^{\frac{\gamma z^2+1}{z^3}}T_1(\alpha,-\beta,\gamma;-1/z),
\]
where the corresponding series defining $T_1$ converges uniformly provided that $z\neq 0$. Then, the Final Value Theorem for the Z-transform \cite{Elaydi} implies that
\[
\lim_{n\to\infty}w_n=\lim_{z\to 1}(z-1)\widetilde{x}(z)=0.
\]
If $\gamma=0$, we can still use the particular solution $\widetilde{x}_1$ found before together with Abel's formula. Let $\widetilde{x}_2$ denote another particular solution 
of the triconfluent Heun equation arising from (\ref{equazza}) after the transformation $\omega=1/z$. Then, the Wronskian is given by $W(T_1(\alpha,\beta,0;\omega);\widetilde{x}_2(\omega))=C\mbox{exp}(\omega^3)$ for some $C\neq 0$. Employing the definition of the Wronskian we end up with the following first order, linear differential equation for 
$\widetilde{x}_2$, namely
\[
\frac{d\widetilde{x}_2}{d\omega}-\frac{T_1^{'}(\alpha,\beta,0;\omega)}{T_1(\alpha,\beta,0;\omega)}\widetilde{x}_2=C\frac{e^{\omega^3}}{T_1(\alpha,\beta,0;\omega)}.
\]
The corresponding integrating factor is given by $\mu(\omega)=T_1^{-1}(\alpha,\beta,0;\omega)$ and hence
\[
\widetilde{x}_2(\omega)=T_1(\alpha,\beta,0;\omega)\int\frac{e^{\omega^3}}{T_1^2(\alpha,\beta,0;\omega)}~d\omega.
\]
Finally, the general solution of (\ref{ascia}) for $\gamma=0$ is represented by
\[
\widetilde{x}(z)=T_1(\alpha,\beta,0;1/z)\left(D_1 +D_2\int\frac{e^{1/z^3}}{z^2T_1^2(\alpha,\beta,0;1/z)}~dz\right).
\]
Taking into account that the following expansion holds for $|z-1|<1$
\[
\int\frac{e^{1/z^3}}{z^2T_1^2(\alpha,\beta,0;1/z)}~dz=\frac{e(z-1)}{T_1^2(\alpha,\beta,0;1)}+\mathcal{O}(z-1)^2,
\]
once again the Final Value Theorem implies that $\lim_{z\to 1}(z-1)\widetilde{x}(z)=0$ and this completes the proof.~~$\square$
\end{proof}
We conclude the analysis of the recurrence relation (\ref{cuadro}) by constructing asymptotic expansions valid as $n\to\infty$ consisting of an exponential leading term 
multiplied by a descending series. These kind of expansions are called Birkhoff series \cite{Birkhoff1,Birkhoff2}. To construct such series we will assume as in \cite{Wimp} that 
\begin{equation}\label{w1}
w_n=\mathcal{E}_n K_n,\quad 
\mathcal{E}_n=e^{\mu_0 n\ln{n}+\mu_1 n}n^\theta,\quad 
K_n=e^{\alpha_1 n^{\widetilde{\beta}}+\alpha_2 n^{\widetilde{\beta}-1/\widetilde{\rho}}+\cdots}
\end{equation}
with $\alpha_1\neq 0$, $\widetilde{\beta}=j/\widetilde{\rho}$, $0\leq j<\rho$, $\rho\geq 1$ and $\mu_0\rho$ taking integer values. Then, for $k=1,2,3$ we have
\begin{equation}\label{w2}
\frac{w_{n+k}}{w_n}=n^{\mu_0 k}\lambda^k\left[1+\frac{A}{n}+\frac{B}{n^2}+\frac{C}{n^3}+\cdots\right]e^{\alpha_1\widetilde{\beta}k n^{\widetilde{\beta}-1}+
\alpha_2\left(\widetilde{\beta}-\frac{1}{\widetilde{\rho}}\right)kn^{\widetilde{\beta}-(1/\widetilde{\rho})-1}+\cdots},\quad 
\lambda=e^{\mu_0+\mu_1}.
\end{equation}
where
\begin{eqnarray*}
A&=&k\theta +\frac{\mu_0 k^2}{2},\\
B&=&\frac{\mu_0^2 k^4}{8}-\frac{\mu_0 k^3}{6}+\frac{1}{2}\mu_0 k^3\theta+\frac{k^2}{2}\theta(\theta-1),\\
C&=&\frac{\mu_0^3 k^6}{48}-\frac{\mu_0^2 k^5}{12}+\frac{\mu_0 k^4}{12}+\theta k\left(\frac{\mu_0^2 k^4}{8}-\frac{\mu_0 k^3}{6}\right)
+\frac{\mu_0 k^4}{4}\theta(\theta-1)+\frac{k^3}{6}\theta(\theta-1)(\theta-2).
\end{eqnarray*}
To construct Birkhoff series for the recurrence relation (\ref{cuadro}) it turns out to be convenient to rewrite it as
\[
(\beta-3-3n)w_n+\alpha w_{n+1}-\gamma(n+2)w_{n+2}+(n+2)(n+3)w_{n+3}=0.
\]
Substituting (\ref{w1}) and (\ref{w2}) into the above equation, we find the equation to be formally satisfied is
\[
\beta-3-3n+\alpha\lambda n^{\mu_0}\left[1+\frac{\theta+\mu_0/2}{n}+\mathcal{O}(n^{-2})\right]e^{\alpha_1\widetilde{\beta} n^{\widetilde{\beta}-1}+
\alpha_2\left(\widetilde{\beta}-\frac{1}{\widetilde{\rho}}\right)n^{\widetilde{\beta}-(1/\widetilde{\rho})-1}+\cdots}
\]
\[
-\gamma\lambda^2 n^{2\mu_0+1}\left[1+\frac{2(\theta+\mu_0+1)}{n}+\mathcal{O}(n^{-2})\right]e^{2\alpha_1\widetilde{\beta} n^{\widetilde{\beta}-1}+
2\alpha_2\left(\widetilde{\beta}-\frac{1}{\widetilde{\rho}}\right)n^{\widetilde{\beta}-(1/\widetilde{\rho})-1}+\cdots}
\]
\[
+\lambda^3 n^{3\mu_0+2}\left[1+\frac{3\theta+\frac{9}{2}\mu_0+5}{n}+\mathcal{O}(n^{-2})\right]e^{3\alpha_1\widetilde{\beta} n^{\widetilde{\beta}-1}+
3\alpha_2\left(\widetilde{\beta}-\frac{1}{\widetilde{\rho}}\right)n^{\widetilde{\beta}-(1/\widetilde{\rho})-1}+\cdots}=0.
\]
Obviously, it must be $\mu_0=-1/3$ and this implies $\lambda=3^{1/3}$. This further requires that $\widetilde{\rho}=3$ and $\widetilde{\beta}=1/3$. Expanding the exponentials 
gives
\[
(3\alpha_1-3^{2/3}\gamma)n^{1/3}+\beta-3+3\left(3\theta+\frac{7}{2}\right)+(\mbox{lower order terms})=0.
\]
Thus, $\alpha_1=\gamma/3^{1/3}$ and $\theta=-5/6-\beta/9$. Hence, the first asymptotic series is
\[
w^{(1)}_n\sim\left(\frac{3e}{n}\right)^{n/3}n^{-\frac{5}{6}-\frac{\beta}{9}}e^{\gamma\sqrt[3]{\frac{n}{3}}}
\left[1+\frac{c_1}{n^{1/3}}+\frac{c_2}{n^{2/3}}+\frac{c_3}{n}+\mathcal{O}(n^{-4/3})\right]
\]
with
\begin{eqnarray*}
c_1&=&\frac{1}{\sqrt[3]{9}}\left(\alpha-\frac{\gamma^2}{6}\right),\\
c_2&=&\frac{1}{2\sqrt[3]{9}}\left[\gamma\left(1-\frac{\beta}{3}\right)+\left(\alpha-\frac{\gamma^2}{6}\right)^2\right],\\
c_3&=&\frac{1}{324}\left[-108 c_1^3+324c_1 c_2+(8\sqrt[3]{1089}-36\sqrt[3]{9})\gamma c_1+6\beta^2+18\beta-2\gamma^3-81\right].
\end{eqnarray*}
According to \cite{Wimp} other two formal series solutions can be obtained by letting $n\to ne^{2\pi k i}$ with $k=0,1,\cdots,\widetilde{\rho}-1$. In the present case we have 
$\widetilde{\rho}=3$ and the corresponding Birkhoff series are thus
\begin{eqnarray*}
w^{(2)}_n&\sim&\left(\frac{3e}{n}\right)^{n/3}n^{-\frac{5}{6}-\frac{\beta}{9}}e^{-\frac{2}{3}\pi i n-\frac{1}{2}(1-i\sqrt{3})\gamma\sqrt[3]{\frac{n}{3}}}
\left[1-\frac{(1+i\sqrt{3})c_1}{2n^{1/3}}-\frac{(1-i\sqrt{3})c_2}{2n^{2/3}}+\frac{c_3}{n}+\mathcal{O}(n^{-4/3})\right],\\
w^{(3)}_n&\sim&\left(\frac{3e}{n}\right)^{n/3}n^{-\frac{5}{6}-\frac{\beta}{9}}e^{-\frac{4}{3}\pi i n-\frac{1}{2}(1+i\sqrt{3})\gamma\sqrt[3]{\frac{n}{3}}}
\left[1-\frac{(1-i\sqrt{3})c_1}{2n^{1/3}}-\frac{(1+i\sqrt{3})c_2}{2n^{2/3}}+\frac{c_3}{n}+\mathcal{O}(n^{-4/3})\right].
\end{eqnarray*}
Last but not least, note that the above asymptotic results for $w_n^{(i)}$ with $i=1,2,3$ tend to zero for $n\to\infty$ as we would expect from Theorem~\ref{TM}.

\section{Analysis of the triconfluent Heun equation using supersymmetric quantum mechanics}
For the following analysis it is convenient to rewrite the triconfluent Heun equation (\ref{TCHeq}) as $(L_T v)(\rho)=0$ where
\[
L_T=-\frac{d^2}{d\rho^2}+\Omega(\rho),\quad\Omega(\rho)=-A_0-A_1\rho-A_2\rho^2+\frac{9}{4}\rho^4.
\]
We look for a factorization of the operator $L_T$ having the form 
\[
L_T=AA^\dagger=\left(\frac{d}{d\rho}+W(\rho)\right)\left(-\frac{d}{d\rho}+W(\rho)\right).
\]
This will be the case if the unknown function $W$ satisfies the Riccati equation
\[
\frac{dW}{d\rho}+W^2(\rho)=\Omega(\rho).
\]
We solved the above equation with the help of the software package Maple~$18$ and we found that the superpotential $W$ is given by
\[
W(\rho)=\frac{1}{6}\frac{e^{f(\rho)}F(\alpha,\beta,\gamma;\rho)-c_1 e^{-f(\rho)}F(\alpha,-\beta,\gamma;-\rho)}{e^{f(\rho)}
T_1(\alpha,\beta,\gamma;\rho)+c_1 e^{-f(\rho)}T_1(\alpha,-\beta,\gamma;-\rho)}
\]
where $c_1$ is an arbitrary integration constant and 
\[
f(\rho)=-\frac{\gamma}{2}\rho-\frac{\rho^3}{2},\quad
F(\alpha,\beta,\gamma;\rho)=-3(\gamma+3\rho^2)T_1(\alpha,\beta,\gamma;\rho)+6\dot{T}_1(\alpha,\beta,\gamma;\rho).
\]
with $\alpha$, $\beta$, and $\gamma$ as in (\ref{abc}). Note that as a byproduct result we also obtained a new factorization of the triconfluent Heun equation since it does not coincide with the one offered in \cite{Ronveaux}.
\begin{figure}\label{FW}
\includegraphics[scale=0.4]{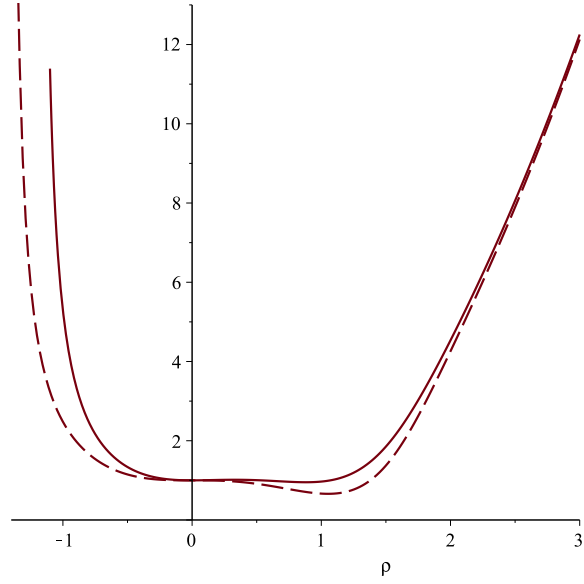}
\caption{\label{figW}
Behaviour of the superpotential $W$ for different values of the parameters: solid line corresponds to $\alpha=0=c_1$, $\beta=-1$, $\gamma=-2$; dash line $\alpha=0=c_1=\beta$, $\gamma=-2$.
}
\end{figure}

Combining the operators $A$ and $A^\dagger$ we can construct Hamiltonians
\[
\mathcal{H}_{-}=A^\dagger A=-\frac{d^2}{d\rho^2}+V_{-}(\rho),\quad 
\mathcal{H}_{+}=AA^\dagger=L_T=-\frac{d^2}{d\rho^2}+V_{+}(\rho)
\]
where $V_{+}$ and $V_{-}$ are the supersymmetric partner potentials and they are given by
\[
V_{\pm}(\rho)=W^2(\rho)\pm\frac{dW}{d\rho}.
\]
Clearly, $V_{+}(\rho)=\Omega(\rho)$ and moreover the other partner potential can be expressed as
\[
V_{-}(\rho)=2W^2(\rho)-\Omega(\rho)=\Omega(\rho)-2\frac{dW}{d\rho}.
\]
\begin{figure}\label{FW2}
\includegraphics[scale=0.4]{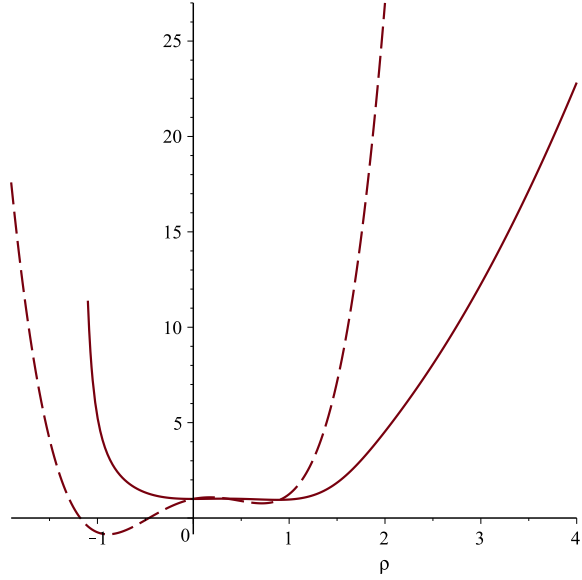}
\caption{\label{figW2}
Behaviour of the partner potentials $V_{-}$ (solid line) and $V_{+}$ (dash line) for $\alpha=0$, $\beta=-1$, $\gamma=-2$, $c_1=0$.
}
\end{figure}
If we denote the eigenstates of the Hamiltonians $\mathcal{H}_{-}$ and $\mathcal{H}_{+}$ by $v_n^{(-)}$ and $v_n^{(+)}$, respectively, then $v_n^{(\pm)}$ must satisfy the eigenvalue equations
\[
\mathcal{H}_{-}v_n^{(-)}=E^{(-)}_n v_n^{(-)},\quad \mathcal{H}_{+}v_n^{(+)}=E^{(+)}_n v_n^{(+)}.
\]
In order to establish whether or not the supersymmetry is unbroken we need to find out if the ground state of $\mathcal{H}_{-}$ has zero energy, that is $E^{(-)}_0=0$.

\section{Analysis of the zero energy state}
Let $\psi_0$ denote the zero energy solution, i.e. $E=0$, associated to an Hamiltonian $H_0$ such that $H_0\psi_0=0$ and 
$H_0=d^2/dr^2+V_{eff}(r)$ where the effective potential has been already computed in Section~\ref{red}. For the potential $V_3$ with 
$e_0=e_1=e_2=0$ corresponding to the choice $a_0=a_1=a_2=0$ we find that the zero energy solution can be written as a linear combination of Bessel functions as follows
\[
\psi_0(r)=\sqrt{r}\left[C_1 J_{\sqrt{7}/6}\left(\frac{2\sqrt{e_4}}{3}r^{3/2}\right)+
C_2 Y_{\sqrt{7}/6}\left(\frac{2\sqrt{e_4}}{3}r^{3/2}\right)\right]
\]  
with integration constants $C_1$ and $C_2$. This solution is not square integrable since it oscillates asymptotically as a plane wave. 
Therefore, it is obvious that we are handling here the scattering solution for a spacial case of the energy.
Note that from the definition of $e_4$ we can conclude that it is always strictly positive.  Another case that can be solved exactly is the following. Consider the potential $V_4$ and let $a_0=a_2=b_0=0$. Then, the coefficients $v_1$, $v_3$, and $v_4$ vanishes whereas $v_0=-a1/b_1$. Hence, we have
\[
V_4(r)=-\frac{5}{36 r^2}+v_5 r^2+v_0
\] 
with $v_5>0$. For the above potential the Schr\"{o}dinger equation
\[
\frac{d^2\psi}{dr^2}+V_4(r)\psi(r)=E\psi(r)
\]
can be exactly solved in terms of a linear combination of Whittaker functions given by
\[
\psi_E(r)=\frac{1}{\sqrt{r}}\left[C_1 M_{\alpha,\beta}(i\sqrt{v_5}r^2)+C_2 W_{\alpha,\beta}(i\sqrt{v_5}r^2)\right]
\]
with
\[
\alpha=\frac{i(E-v_0)}{4\sqrt{5}},\quad\beta=\frac{\sqrt{14}}{12}.
\]
Observe that if $E=v_0$ or $E=0$ and $v_0=0$ we find that 
\[
\psi_{E=v_0}(r)=\sqrt{r}\left[C_1 J_{\beta}\left(\frac{\sqrt{v_5}}{2}r^2\right)+
C_2 Y_{\beta}\left(\frac{\sqrt{v_5}}{2}r^2\right)\right].
\]

\end{document}